\documentclass[a4paper,10pt]{article}
\usepackage{amssymb,amsmath,amsthm}
\usepackage{fullpage}
\usepackage{hyperref}
\usepackage{eucal}
\usepackage{color}
\usepackage{stackrel}
\usepackage{graphicx}
\newcommand{\ie}{\emph{i.e.}}
\newcommand{\eg}{\emph{e.g.}}
\newcommand{\cf}{\emph{cf.}}
\newcommand{\etal}{\emph{et al.}}

\newcommand{\Real}{\mathbb{R}}

\newcommand{\Sphere}{\mathbb{S}}

\newcommand{\supp}{\mathop{\mathrm{supp}}\nolimits}

\newcommand{\dom}{\mathop{\mathrm{dom}}\nolimits}
\newcommand{\dist}{\mathop{\mathrm{dist}}\nolimits}

\newcommand{\tr}{\mathop{\mathrm{tr}}\nolimits}
\newcommand{\Cut}{\mathrm{Cut}}
\newcommand{\eps}{\varepsilon}
\newcommand{\sii}{L^2}

\newcommand{\der}{\mathrm{d}}

\newtheorem{Theorem}{Theorem}
\newtheorem{Lemma}{Lemma}

\newtheorem{Conjecture}{Conjecture}
\theoremstyle{definition}
\newtheorem{Remark}{Remark}

\begin{document}
%
%
%-------%
% TITLE %
%-------%
%------------------------------------------%
%------------------------------------------%
\title{\Large\textbf{%
Bound states in soft quantum layers}}%
\author{David Krej\v{c}i\v{r}{\'\i}k\,$^a$ \ and \ Jan K\v{r}\'i\v{z}\,$^b$}	
\date{\small
\begin{quote}
\emph{
\begin{itemize}
\item[$a)$]
Department of Mathematics, Faculty of Nuclear Sciences and
Physical Engineering, Czech Technical University in Prague,
Trojanova 13, 12000 Prague 2, Czechia;
david.krejcirik@fjfi.cvut.cz.%
\item[$b)$]
Department of Physics, Faculty of Science, University of Hradec Kr\'alov\'e,
Rokitansk\'eho 62, 500 03 Hradec Kr\'alov\'e, Czechia;
jan.kriz@uhk.cz
\end{itemize}
}
\end{quote}
13 April 2023
\bigskip \\
\fbox{
\textbf{To appear in:} \
\emph{Publ. RIMS, Kyoto University}
}
}
\maketitle

\begin{abstract}
\noindent
We develop a general approach to study three-dimensional
Schr\"odinger operators with confining potentials depending on
the distance to a surface.
The main idea is to apply parallel coordinates based on the surface 
but outside its cut locus in the Euclidean space.
If the surface is asymptotically planar
in a suitable sense,
we give an estimate on the location of the essential 
spectrum of the Schr\"odinger operator.
Moreover, if the surface  coincides 
up to a compact subset with a surface of revolution
with strictly positive total Gauss curvature,
it is shown that the Schr\"odinger operator 
possesses an infinite number of discrete eigenvalues.
%
%\bigskip
%\begin{itemize}
%\item[\textbf{Keywords:}]
%\item[\textbf{MSC (2010):}]
%\end{itemize}
%
\end{abstract}
%
%------------------------------------------%
%------------------------------------------%

%---------------------%
\section{Introduction}
%---------------------%
%
Consider a non-relativistic quantum particle propagating in the vicinity
of an unbounded surface~$\Sigma$ in~$\Real^3$. 
Spectral properties of the \emph{hard-wall} idealisation,
where the Hamiltonian is identified with the Dirichlet Laplacian
in the tubular neighbourhood called \emph{layer}
\begin{equation}\label{tube}
  \Omega_a := \{x\in\Real^3 : \ \dist(x,\Sigma)<a\}
  \,,
\end{equation}
were first analysed by Duclos \etal~in the pioneering work~\cite{DEK2}. 
While the essential spectrum is stable under local perturbations 
of the straight layer $\Real^2 \times (-a,a)$,
the most interesting result of the study is 
the existence of \emph{bound states},
\ie\ discrete eigenvalues.
This highly non-trivial property for unbounded domains 
was established in~\cite{DEK2} under rather restrictive 
geometric and topological conditions about~$\Sigma$.
However, the subsequent works of Carron \etal~\cite{CEK} 
and notably of Lu \etal\ \cite{LL1,LL2,LL3,Lu-Rowlett_2012,KL}
have demonstrated that the existence of discrete spectra
due to bending is indeed a robust phenomenon.
See also \cite{EK3,BEGK,KRT,KT2,FK8,Exner-Lotoreichik_2020} 
for quantitative properties of the eigenvalues and eigenfunctions,
and \cite{Exner-Tater_2010,Dauge-Ourmieres-Bonafos-Raymond_2015,
Dauge-Lafranche-Ourmieres-Bonafos_2018,
Ourmieres-Pankrashkin_2018,KLO} 
for layers over non-smooth surfaces.

To allow for quantum tunnelling,
Exner and Kondej in~\cite{ESylwia2} introduced 
a \emph{leaky} realisation of the confinement to~$\Sigma$
by considering the singular Schr\"odinger operator 
$-\Delta + \alpha \delta_\Sigma$ in $\sii(\Real^3)$,
where~$\delta_\Sigma$ is the Dirac delta function and $\alpha< 0$.
Under suitable geometric assumptions about the surface~$\Sigma$,
the authors demonstrated the existence of discrete spectra 
in the regime of \emph{large confinement}, \ie\ $\alpha \to -\infty$.
The robust existence of the discrete eigenvalues for \emph{all} negative~$\alpha$
is stated as an open problem in \cite[Sec.~7.5]{Exner_2008}.
Spectral analysis of related models can be found in
\cite{Behrndt-Exner-Lotoreichik_2014,
Behrndt-Exner-Lotoreichik_2014b,
Behrndt-Grubb-Langer-Lotoreichik_2015, 
Exner-Rohleder_2016,
Behrndt-Exner-Holzmann-Lotoreichik_2017}.

The purpose of the present paper is to investigate 
the existence of discrete spectra in yet another realisation of 
the confinement, namely when the particle Hamiltonian is identified with 
the Schr\"odinger operator
\begin{equation}\label{Schrodinger}
  H := -\Delta + V 
  \qquad \mbox{in} \qquad
  \sii(\Real^3)
  \,,
\end{equation}
where~$V$ is a \emph{regular} potential modelling a force which constrains
the particle to the tubular neighbourhood~$\Omega_a$.
Extending the terminology of Exner~\cite{Exner_2020,Exner_2022}
for analogous models when the submanifold is a curve
to the present case of surfaces,
we call these realisations \emph{soft layers}. 
In this case, there exists only a general asymptotic spectral analysis
by adiabatic methods
of Wachsmuth and Teufel \cite{Wachsmuth-Teufel_2013}
(see also \cite{Haag-Lampart-Teufel_2015}),
from which it follows that the discrete spectrum 
will exist for \emph{deep and narrow} confining potentials~$V$
(in agreement with the leaky model above).

From a more general perspective, the hard-wall and leaky realisations
fall into the unifying scheme~\eqref{Schrodinger} 
provided that we \emph{formally} set
$V_\mathrm{hard} := \infty \chi_{\Real^3\setminus\Omega_a}$ 
and $V_\mathrm{leaky} := \alpha \delta_\Sigma$.
This can be made rigorous by considering the Dirichlet boundary
conditions for the Laplacian in $\sii(\Omega_a)$ 
or by defining~\eqref{Schrodinger} by means of 
the sesquilinear form, respectively. 
As a matter of fact, the present approach yields new results
for the leaky layers too, namely the robust existence 
of discrete eigenvalues for \emph{all} negative~$\alpha$,
solving in this way the open problem of \cite[Sec.~7.5]{Exner_2008},
at least in a special class of rotationally symmetric geometries. 
The hard-wall layers could be also treated simultaneously, 
but our technique does not bring anything new in this case
(except for the explicit observation missing in~\cite{DEK2,CEK}
that there is an \emph{infinite} number of eigenvalues
in hard-wall layers with appropriate rotationally symmetric ends).

Before stating our main results, 
let us informally summarise the characteristic hypotheses. 
The surface~$\Sigma$ is assumed to be smooth and orientable,
the Gauss curvature of~$\Sigma$ is integrable (see~\eqref{Ass.K}) 
and~$\Sigma$ is \emph{asymptotically planar}
in the sense that both the Gauss and mean
curvatures vanish at infinity of~$\Sigma$ (see~\eqref{Ass.planar}). 
More restrictively, we assume that~$\Sigma$ 
is \emph{asymptotically cut-locus planar} 
(see~\eqref{Ass.planar.cut}).
As usual in the theory of quantum waveguides,
we always assume that the tubular neighbourhood~\eqref{tube}
does not overlap itself with some positive~$a$
(see~\eqref{Ass.local} and~\eqref{Ass.global}). 
Finally, $\Sigma$~is assumed to contain a cylindrically
symmetric end with positive total Gauss curvature
(\ie~the integral of the Gauss curvature is positive),
whose asymptotic cut-locus is known explicitly (see~\eqref{Ass.cut})
and whose parallel curvature admits a power-like decay (see~\eqref{Ass.decay}). 

The confining potential $V:\Real^3 \to \Real$ is assumed 
to be an essentially bounded function 
or the leaky realisation $V_\mathrm{leaky} := \alpha \delta_\Sigma$
with $\alpha \in \Real$.
In the former case, we assume that the support of~$V$ 
is contained in the closure of~$\Omega_a$  
and that the profile does not vary along~$\Sigma$.
More specifically, if~$n$ is a unit normal vector field along~$\Sigma$
and $p \in \Sigma$,
we assume 
\begin{equation}\label{Ass.vary}
  \fbox{
  \mbox{$W(t) := V\big(p+n(p)\,t\big) $ is independent of~$p$}
  \qquad \& \qquad
  $\supp W \subset [-a,a]$
  .
  }
\end{equation}
This is certainly the case of leaky layers too,
because~$\delta_\Sigma$ is zero range
and~$\alpha$ is assumed to be a constant.
 The corresponding one-dimensional operator
$$
  T := -\partial_t^2 + W(t)
  \qquad \mbox{in} \qquad
  \sii(\Real)
$$
with form domain~$H^1(\Real)$
(the sum should be understood as the form sum in the leaky case) 
has the essential spectrum covering $[0,\infty)$ in both cases.
We assume that~$W$ is attractive in the sense that
\begin{equation}\label{Ass.attractive}
  \fbox{
  \mbox{$T$ possesses at least one negative eigenvalue.
  }}
\end{equation}
This hypothesis holds in the leaky case
if, and only if, $\alpha$~is a negative constant.
In general, a sufficient condition to guarantee~\eqref{Ass.attractive}
is that $\int_\Real W\, < 0$
(which particularly involves negative potentials).
Moreover, it is easy to design potentials which simultaneously
satisfy $\int_\Real W\, \geq  0$ and~\eqref{Ass.attractive}
(\eg, it is enough to consider the strong coupling regime
of any~$W$ possessing a negative minimum, see~\cite[Thm.~4]{FK1}).
Let $E_1<0$ denote the lowest discrete eigenvalue of~$T$.

Our main result reads as follows.
\begin{Theorem}\label{Thm.main} 
Let~$\Sigma$ be an orientable smooth surface
which is asymptotically 
cut-locus planar~\eqref{Ass.planar.cut}
and admits an integrable Gauss curvature~\eqref{Ass.K}.
Let the tubular neighbourhood~\eqref{tube}
do not overlap itself with some positive~$a$,
\ie~\eqref{Ass.local} and~\eqref{Ass.global} hold. 
Let~$V$ be an essentially bounded function 
(or the distribution $\alpha \delta_\Sigma$ 
with $\alpha < 0$) 
satisfying~\eqref{Ass.vary} and~\eqref{Ass.attractive}.
Then 
$$
  \inf\sigma_\mathrm{ess}(H) \geq E_1 \,. 
$$
Moreover, if~$\Sigma$ coincides up to a compact subset
with a cylindrically symmetric surface
with positive total Gauss curvature
and satisfying the extra hypotheses~\eqref{Ass.cut} and~\eqref{Ass.decay},   
then~$H$ possesses an infinite number 
(counting multiplicities) of discrete eigenvalues below~$E_1$,
and in this case $\inf\sigma_\mathrm{ess}(H) = E_1$.
\end{Theorem}

A special circumstance is of course when~$\Sigma$
coincides with the cylindrically symmetric end. 
Then a canonical example of the surface satisfying all the hypotheses
is the paraboloid of revolution 
(and other surfaces obtained by revolving polynomially growing curves).
Another typical example is the family of surfaces~$\Sigma_\theta$ 
obtained be revolving the planar curve of~\cite{KKK}
(see Figure~\ref{Fig.KKK})
\begin{equation}\label{curve.special} 
  \Gamma_\theta(s) :=
  \begin{cases}
    \left(
    R \sin \frac{s}{R},
    R (1-\cos \frac{s}{R})
    \right)
    & \mbox{if} \quad  s \in [0,\frac{\theta}{2} R)
    \,,
    \\
    \left(
    (s-\frac{\theta}{2} R) \cos\frac{\theta}{2} + R \sin\frac{\theta}{2},
    (s-\frac{\theta}{2} R) \sin\frac{\theta}{2}
    + R (1-\cos\frac{\theta}{2})
    \right)
    & \mbox{if} \quad s \in [\frac{\theta}{2} R,\infty)
    \,,
  \end{cases}
\end{equation}
{along the second axis in~$\Real^3$,
where $R>0$ and $\theta \in [0,\pi]$.
Hence~$\Sigma_\theta$ is a union of a spherical cap and a conical end;
it is a plane if $\theta=0$, while the end becomes cylindrical if $\theta=\pi$. 
Of course, $\Sigma_\theta$ is not smooth (unless $\theta=0$),
but it is piecewise smooth (in fact, piecewise analytic).

\begin{figure}[h]
\begin{center}
\begin{tabular}{ccc}
\includegraphics[width=.50\textwidth]{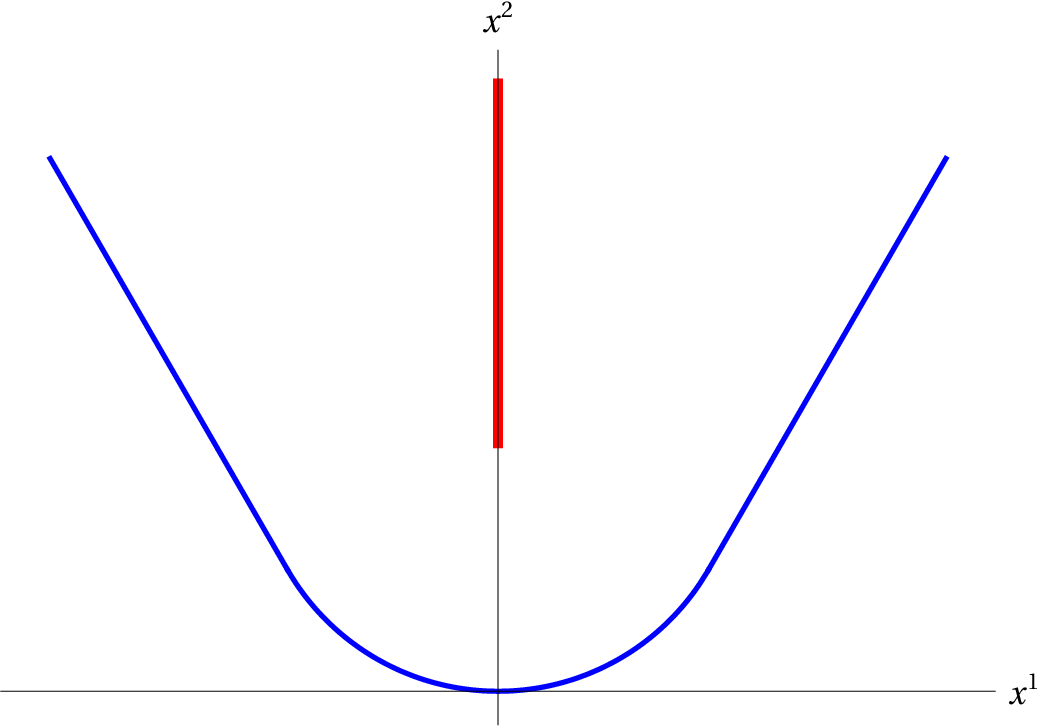}
&& \includegraphics[width=.32\textwidth]{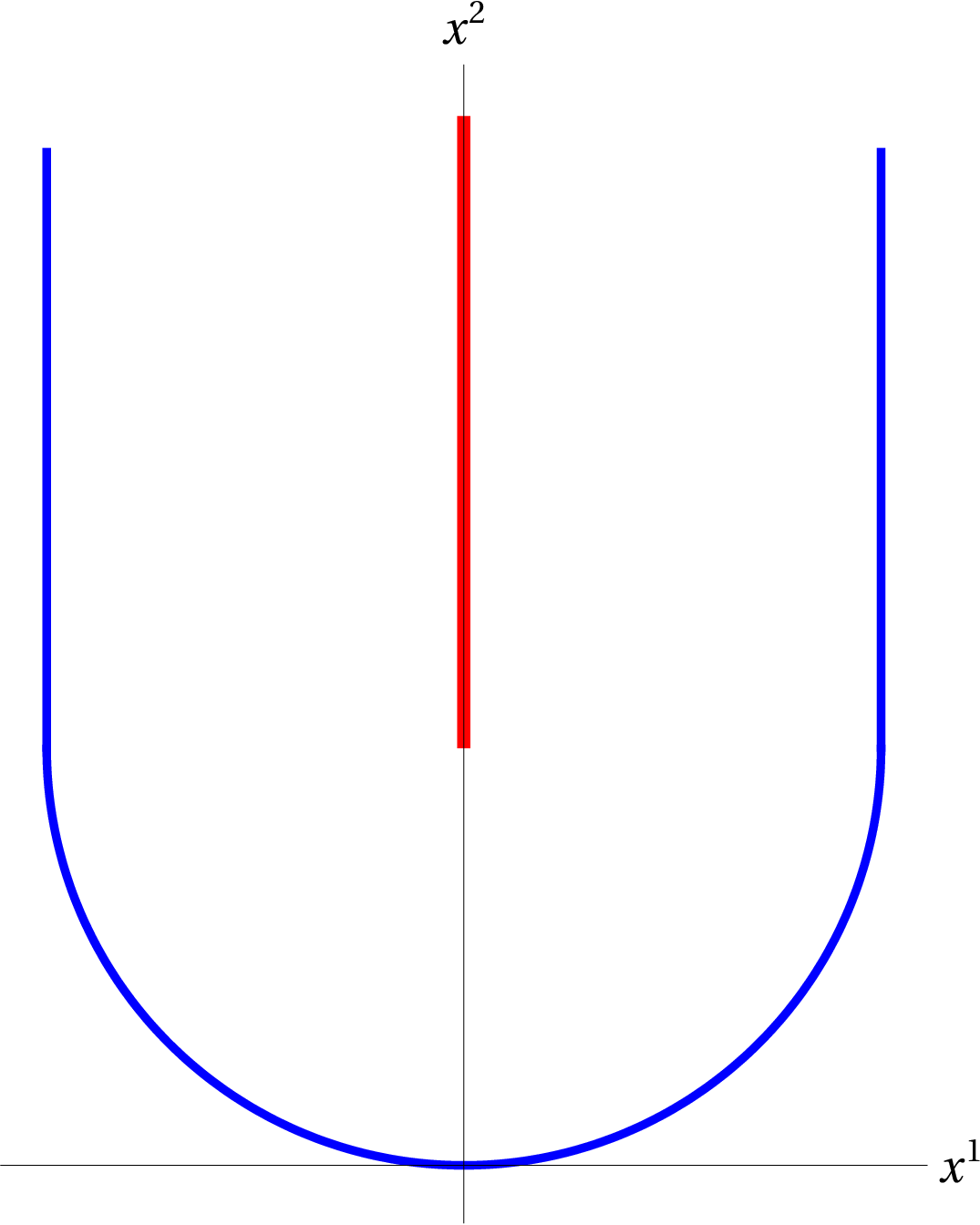}
\\
$\theta \in (0,\pi)$ 
&& $\theta=\pi$
\end{tabular}
\end{center}
\caption{The piecewise smooth curve~\eqref{curve.special}
(symmetrically extended to $s \in \Real$)
and its cut-locus (red).} 
\label{Fig.KKK}
\end{figure}

The surface~$\Sigma_\theta$ can be regarded as a regularised version
of the conical geometry considered in 
\cite{Exner-Tater_2010,Dauge-Ourmieres-Bonafos-Raymond_2015,
Ourmieres-Pankrashkin_2018,Egger-Kerner-Pankrashkin_2020}. 
Indeed, Theorem~\ref{Thm.main} applies to~$\Sigma_\theta$
with $\theta \in (0,\pi)$, confirming in this way 
the results of the precedent works.
In fact, not necessarily rotationally symmetric cones are considered in 
\cite{Ourmieres-Pankrashkin_2018,Egger-Kerner-Pankrashkin_2020},
and moreover the accumulation rate of the eigenvalues is derived there. 
On the other hand, the strength of the present work is 
that we go substantially beyond the conical geometries,
so the present paper can be considered as a generalisation of
\cite{Ourmieres-Pankrashkin_2018,Egger-Kerner-Pankrashkin_2020}.

The present paper can be also considered as a generalisation 
of our precedent work~\cite{KKK} to three dimensions.
Indeed, our \emph{modus operandi} is again rooted in developing 
the method of parallel coordinates based on~$\Sigma$  
involving the \emph{cut-locus} of~$\Sigma$. 
Unfortunately, the common limitation of the method is 
that we can consider special surfaces only:
those for which the cut-locus is known explicitly.
However, the unprecedented novelty with respect to~\cite{KKK}
is that an \emph{asymptotic} knowledge of the cut-locus
is enough in the present work. 
This enables us to cover more general geometries.
What is more,  the present proof of the existence of discrete spectrum 
exhibits another important novelty with respect to the previous work:
The argument requires a careful choice of 
the trial function \emph{localised at infinity}.
This is indeed a huge difference with respect to~\cite{KKK}, 
where the trial function was the standard one,
\ie~essentially a constant localised everywhere. 
This phenomenon is closely related to 
the existence of the \emph{intrinsic} Gauss curvature~$K$
for surfaces, while there are just extrinsic curvatures for curves. 
It has been noticed already in~\cite{DEK2} that
a more refined choice of trial functions is necessary
for layers over surfaces with \emph{positive} total Gauss curvature.
Furthermore, the unconventional choice of the trial function
enables one to conclude that there is actually 
an \emph{infinite} number of discrete eigenvalues. 

The paper is organised as follows. 
In Section~\ref{Sec.parallel} we develop a general approach
to soft and leaky quantum layers, in particular we introduce
a useful parameterisation of~$\Real^3$ involving 
the cut-locus of~$\Sigma$.
In Section~\ref{Sec.symmetry} we consider the special situation
of cylindrically symmetric layers.
Theorem~\ref{Thm.main} is proved in Section~\ref{Sec.proof}.
 
%-----------------------------% 
\section{Parallel coordinates}\label{Sec.parallel}
%-----------------------------% 
% 
Let us first develop a general approach to study soft quantum layers. 
The first part of our approach (before speaking about the \emph{cut-locus})
is rather standard and we refer to~\cite{CEK,KRT} 
for similar geometric preliminaries.

Let~$\Sigma$ be a connected orientable smooth surface in~$\Real^3$. 
We are particularly interested in non-compact complete surfaces,
but~$\Sigma$ can be alternatively compact in these geometric preliminaries.   
The induced metric of~$\Sigma$ will be denoted by~$g$.
Introducing the standard notation $|g| := \det(g)$,
the surface element of~$\Sigma$ reads 
$\der\Sigma = |g|^{1/2} \, \der s^1 \wedge \der s^2$,
where~$(s^1,s^2)$ is a local coordinate system of~$\Sigma$.
The orientation of~$\Sigma$ is specified by a globally defined 
unit normal vector field $n \in \Sigma \to \Sphere^2$.
For any $p \in \Sigma$, we introduce the \emph{Weingarten map}
$$
  L:T_p\Sigma \to T_p\Sigma : \{ \xi \mapsto -\der n (\xi)\}
  \,.
$$  
The eigenvalues $k_1,k_2$ of~$L$ are called 
the \emph{principal curvatures} of~$\Sigma$. 
They are defined only locally on~$\Sigma$,
but the \emph{Gauss curvature} $K := \det(L) = k_1 k_2$
and the \emph{mean curvature} $M := \frac{1}{2} \tr(L) = \frac{1}{2}(k_1+k_2)$
are globally defined smooth functions on~$\Sigma$. 
The relationship of~$L$ with the \emph{second fundamental form}~$h$ of~$\Sigma$
is through the formula $L^\mu_{\ \nu} = g^{\mu\rho} h_{\rho\nu}$
where $h = h_{\mu\nu} \der s^\mu \der s^\nu$, 
$g = g_{\mu\nu} \der s^\mu \der s^\nu$
and, as usual, $g^{\mu\nu}$ denote the entries
of the inverse matrix $(g_{\mu\nu})^{-1}$.
Here we adopt the Einstein summation convention,
the range of Greek and Latin indices being $1,2$ and $1,2,3$, respectively.

The characteristic hypothesis of this work is that 
the Gauss curvature is integrable:
\begin{equation}\label{Ass.K}
  \fbox{$
  K \in L^1(\Sigma)
  $.}
\end{equation}
Then the \emph{total Gauss curvature} 
$\mathcal{K} := \int_\Sigma K \, \der\Sigma$ is well defined.
The quantity~$\mathcal{K}$ plays an important role 
in the global geometry of~$\Sigma$.
In fact, by the celebrated Gauss--Bonnet theorem
(see, \eg, \cite[Sec.~6.3]{Kli}),
$\mathcal{K}$~is a topological invariant for closed surfaces.  
By~\cite{Huber}, the hypothesis~\eqref{Ass.K} implies
that~$\Sigma$ is conformally equivalent to a closed surface
from which a finite number of points have been removed.

Let us consider the \emph{normal exponential map}
\begin{equation}\label{layer}
  \Phi: \Sigma \times \Real \to \Real^3 :
  \left\{
  (p,t) \mapsto p + n(p) \, t
  \right\}
  .
\end{equation}
It gives rise to \emph{parallel} (or \emph{Fermi})
``coordinates'' $(p,t)$ based on~$\Sigma$. 
The metric~$G$ induced by~\eqref{layer} 
has a block form
\begin{equation}\label{metric}
  G = g \circ (I-t L)^2 + \der t^2 \,,
\end{equation}
where~$I$ denotes the identity map on~$T_p\Sigma$.
Consequently,
$$
  |G| := \det(G) = |g| \, [\det(I-t L)]^2
  = |g| \, [(1-t k_1)(1-tk_2)]^2
  = |g| \, (1-2Mt+Kt^2)^2
  \,.
$$

The map~$\Phi$ is standardly used in the theory of quantum layers
as a convenient parametrisation of the tubular neighbourhood~$\Omega_a$
introduced in~\eqref{tube}.  
It follows by the inverse function theorem 
that the restricted map $\Phi: \Sigma \times (-a,a) \to \Omega_a$ 
is a local diffeomorphism provided that 
\begin{equation}\label{Ass.local}
  \fbox{$
  0 < a < \big(\max\left\{\|k_1\|_\infty,\|k_2\|_\infty\right\}\big)^{-1}
  $}
\end{equation}
(with the convention that the right hand side equals $\infty$
if the principal curvatures are identically equal to zero).
Of course, to be able to satisfy this inequality 
with a \emph{positive}~$a$, it is necessary to assume 
that the Gauss and mean curvatures are globally bounded functions
(this is automatically satisfied for compact~$\Sigma$).
The crucial requirement that the tubular neighbourhood~\eqref{tube}  
``does not overlap itself'' precisely means that 
$\Phi: \Sigma \times (-a,a) \to \Omega_a$ is a (global) diffeomorphism,
which is ensured by assuming in addition to~\eqref{Ass.local}
the \emph{ad hoc} requirement that 
\begin{equation}\label{Ass.global}
  \fbox{$\Phi \upharpoonright \Sigma \times (-a,a)$ \ is injective.}
\end{equation}
Then $\big(\Sigma\times(-a,a),G\big)$ is an embedded submanifold of~$\Real^3$
and $\Omega_a = \Phi\big(\Sigma \times (-a,a)\big)$ 
has indeed the geometrical meaning of the set of points 
in~$\Real^3$ squeezed between two parallel hypersurfaces at the distance~$a$ 
from~$\Sigma$.
Indeed, within~$\Omega_a$, one observes that $p \mapsto \Phi(p,t)$
is an embedded surface parallel to~$\Sigma$ at the distance~$|t|$ 
for any fixed $t \in (-a,a)$, while $t \mapsto \Phi(p,t)$ is
a straight line (\ie~a geodesic in~$\Real^3$) 
orthogonal to~$\Sigma$ at any fixed point $p \in \Sigma$.  

Now we go beyond the standard approach to quantum layers
by extending the parallel coordinates~$\Phi$ from~$\Omega_a$
to the whole space~$\Real^3$.
We are inspired by \cite[App.~1]{Savo_2001}.
Define the \emph{cut-radius} maps $c_\pm:\Sigma \to (0,\infty]$
by the property that the segment $t \mapsto \Phi(p,t)$
for positive (respectively, negative) $t$
minimises the distance from~$\Sigma$
if, and only if, $t \in [0,c_+(p))$ (respectively, $t \in (-c_-(p),0]$).	
The cut-radius maps are known to be continuous.
The \emph{cut-locus}
\begin{equation}
  \Cut(\Sigma) :=
  \{\Phi(p,c_+(p)) : p \in \Sigma\} \cup
  \{\Phi(p,-c_-(p)) : p \in \Sigma\}
\end{equation}
is a closed subset of~$\Real^3$ of measure zero
(see, \eg, \cite[Chap.~III]{Chavel}).
The map~$\Phi$, when restricted to the set
\begin{equation}
  U := \{ (p,t) \in  \Sigma \times \Real
  : \ -c_-(p) < t < c_+(p)
  \}
\end{equation}
is a diffeomorphism onto $\Phi(U) = \Real^3 \setminus  \Cut(\Gamma)$.
Obviously, one has the inclusion
\begin{equation}\label{conjugate}
  \Cut(\Sigma) \supset \Cut_0(\Sigma)
  := \big\{\Phi(p,t): \ 1-2M(p)\,t+K(p)\,t^2 = 0\big\}
  \,,
\end{equation}
where $\Cut_0(\Sigma)$ is called the \emph{conjugate locus} of~$\Sigma$
(points where the Jacobian of~$\Phi$ vanishes).

Outside the cut-locus, we have the usual coordinates of quantum layers.
If $(s^1,s^2)$ is a local coordinate system of~$\Sigma$,
then $(s^1,s^2,t)$ is a natural local coordinate system of $\Phi(U)$.
With respect to the corresponding coordinate frame, 
the metric~$G$ admits the matrix representation
\begin{equation}\label{metric.bis}
  \big(G_{ij}\big) = 
  \begin{pmatrix}
    \big(G_{\mu\nu}\big) & 0 \\
    0 & 1
  \end{pmatrix}  
  \qquad\mbox{with}\qquad
  G_{\mu\nu} = g_{\mu \rho} (\delta^\rho_\sigma-t L^{\rho}_{\ \sigma})
  (\delta^\sigma_\nu-t L^{\sigma}_{\ \nu}) 
  \,.
\end{equation}
In particular, the volume element of $\Phi(U)$ is given by 
$$
  \der v := (1-2Mt+Kt^2) \, \der\Sigma \wedge \der t
  \,.
$$ 

In agreement with~\cite{DEK2}, 
we say that a non-compact surface~$\Sigma$ 
is \emph{asymptotically planar} if 
the Gauss and mean curvatures vanish at infinity,
which we schematically write as 
\begin{equation}\label{Ass.planar}
  \fbox{$K,M \xrightarrow[]{\ \infty\ }$ 0.}
\end{equation}
Recall that a function~$f$, 
defined on a non-compact manifold~$\Sigma$,
is said to \emph{vanish at infinity} if,
given any positive number~$\eps$,
there exists a compact subset $K \subset \Sigma$   
such that $|f| < \eps$ on $\Sigma \setminus K$.
Similarly, we say that~$f$ \emph{diverges at infinity} 
and write $f \xrightarrow[]{\ \infty\ } \infty$
if, given any positive number~$\eps$,
there exists a compact subset $K \subset \Sigma$   
such that $|f| > \eps^{-1}$ on $\Sigma \setminus K$.
In parallel~\eqref{Ass.planar}, we say that~$\Sigma$
is \emph{asymptotically cut-locus planar} if 
\begin{equation}\label{Ass.planar.cut}
  \fbox{$c_\pm \xrightarrow[]{\ \infty\ } \infty$.}
\end{equation}
Note that~\eqref{Ass.planar.cut} implies~\eqref{Ass.planar}
due to the inclusion~\eqref{conjugate}.
On the other hand, we expect 
that the reverse implication does not hold in general.
\begin{Conjecture}\label{Conj}
There exists a connected surface such that~\eqref{Ass.planar} holds 
but~\eqref{Ass.planar.cut} is violated.
\end{Conjecture}
\noindent
For possibly disconnected surfaces this is obvious
(think about two parallel planes), but constructing 
an explicit connected example seems difficult
(see Figure~\ref{Fig.book} for a partial attempt).
A sufficient condition to ensure hypothesis~\eqref{Ass.planar.cut} 
as a consequence of~\eqref{Ass.planar} is given by 
surfaces of revolutions (\cf~Lemma~\ref{Lem3} below). 

Now we turn from geometric to analytic preliminaries.
Recall our Hamiltonian~$H$ given in~\eqref{Schrodinger}.
If $V : \Real^3 \to \Real$ is an essentially bounded function
(as is indeed the case of the soft realisation of the confinement),
then~$H$ can be introduced as an ordinary operator sum of 
the self-adjoint Laplacian with domain $H^2(\Real^3)$ 
and the maximal operator of multiplication generated by~$V$. 
The associated closed form reads
\begin{equation}\label{form}
  h[u] := \int_{\Real^3} |\nabla u|^2
  + \int_{\Real^3} V |u|^2
  \,, \qquad
  \dom h := H^1(\Real^3)
  \,.
\end{equation}
If~$V$ is the distribution of the leaky type 
$V_\mathrm{leaky} := \alpha \delta_\Sigma$,
the simplest is to start with the form~\eqref{form},
where the second term should be interpreted as
$\alpha \int_{\Sigma} |u|^2$.
Again, it is a well-defined and closed form
under our standing hypothesis~\eqref{Ass.local} and~\eqref{Ass.global}.
In either case, $H$~can be \emph{defined} as the self-adjoint
operator associated with~$h$ (with the properly interpreted second integral)
via the representation theorem \cite[Thm.~VI.2.1]{Kato}.

Finally, we express~$H$ in the parallel coordinates.
This is achieved by means of the unitary map
$
  \mathcal{U} : \sii(\Real^3) \to \sii(U,\der v)
$
defined by $\mathcal{U} u := u \circ \Phi$.
Then $\hat{H} := \mathcal{U} H \mathcal{U}^{-1}$ is the operator associated
with the quadratic form $\hat{h}[\psi] := h[\mathcal{U}^{-1}\psi]$
with $\dom \hat{h} := \mathcal{U} \dom h$.
Explicitly, using the block-diagonal form of the metric~\eqref{metric},
one has 
\begin{equation}\label{form.transformed}
  \hat{h}[\psi]
  = \int_U \overline{\partial_\mu\psi} \,  G^{\mu\nu} \, \partial_\nu\psi
  \, \der v
  + \int_U |\partial_t\psi|^2  \, \der v
  + \int_U W(t) \, |\psi|^2 \, \der v
  \,,
\end{equation}
where $(G^{\mu\nu}) := (G_{\mu\nu})^{-1}$.
Hereafter the last integral should be interpreted as
$\alpha \int_\Sigma |\psi(\cdot,0)|^2 \, \der \Sigma$
in the case of leaky layers.
In the sense of distributions,
the operator~$\hat{H}$ associated with~$\hat{h}$ acts as 
$$
  \hat{H} =
  -|G|^{-1/2} \partial_{\mu} |G|^{1/2} G^{\mu\nu} \partial_{\nu} 
  - |G|^{-1/2} \partial_{t} |G|^{1/2} \partial_{t} 
  + W 
  \,.
$$ 
Here~$W$ is absent in the case of leaky layers,
the influence of the Dirac interaction being realised
by appropriate transmission conditions imposed on~$\Sigma$
in the operator domain.
We shall not need to specify this condition,
for it is enough to work on the level of forms for our purposes.

\begin{figure}[h]
\begin{center}
\begin{tabular}{ccc}
\includegraphics[width=.21\textwidth]{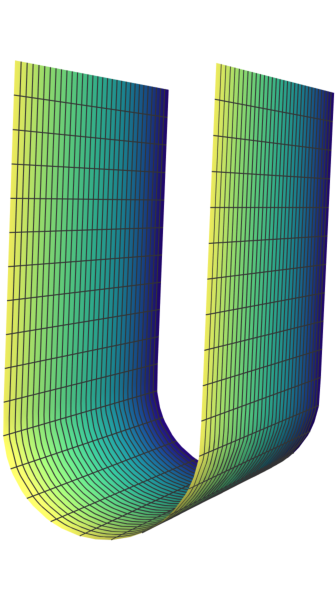}
& \includegraphics[width=.30\textwidth]{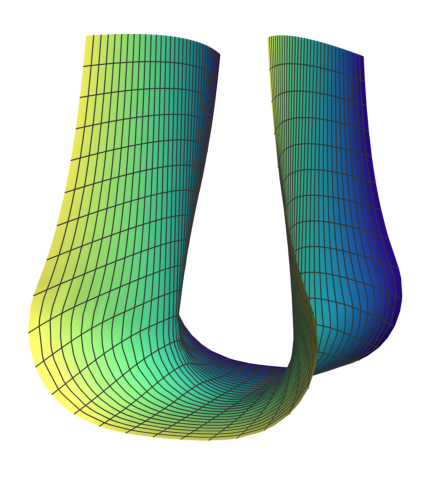}
& \includegraphics[width=.36\textwidth]{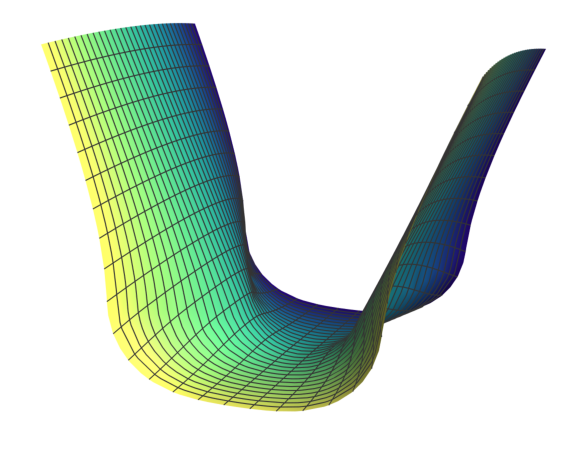}
\\
\begin{tabular}{c}
(a) \\
asymptotically \\
neither cut-locus planar \\
nor planar
\end{tabular}
&
\begin{tabular}{c}
(b) \\
asymptotically \\ planar (along~$x^2$) \\
but not cut-locus planar
\end{tabular}
&
\begin{tabular}{c}
(c) \\
asymptotically \\ cut-locus planar \\ (along~$x^2$)
\end{tabular}
\end{tabular}
\end{center}
\caption{
Towards the proof of Conjecture~\ref{Conj}.
Surface~(a) is 
the curve of Figure~\ref{Fig.KKK} with $\theta=\pi$ and~$R=1$ 
translated along axis~$x^3$.  
Then it is not asymptotically cut-locus planar
because the distance between the flat parts equals the constant~$2$
as $x^2 \to \infty$. Surface~(b) is obtained from~(a)
by taking the radius~$R$ in~\eqref{curve.special} dependent 
both on~$s$ and $x^3$, namely $R(s,x^3) := 1+(x^3)^2/(1+s^2)$.
Then~(b) is not asymptotically cut-locus planar either
(the distance between the modified flat parts remains~$2$
as $x^2 \to \infty$ and $x^3=0$), 
while $K,M \to 0$ as $x^2 \to \infty$ and~$x^3$ is \emph{fixed}
(the challenge is to have~\eqref{Ass.planar} globally).
Surface~(c) is the curve of Figure~\ref{Fig.KKK} 
with $\theta = \frac{5}{7} \pi$ and~$R$ as for~(b); 
then $c_\pm \to \infty$ as $x^2 \to \infty$ and~$x^3$ is fixed.}
\label{Fig.book}
\end{figure}
% 

%--------------------------------------% 
\section{Rotationally symmetric layers}\label{Sec.symmetry}
%--------------------------------------% 
% 
The parallel coordinates of the previous section enables one 
to transfer the geometrically complicated action of the operator~$H$
into the coefficients of the transformed operator~$\hat{H}$.
The problem is that even the form domain $\dom\hat{h}$
is not easy to identify because of
the boundary conditions on~$\partial U$.
An objective of this paper is to point out that
there exists a special class of surfaces 
for which this is feasible
because of a more precise information about the cut-locus.
These are surfaces of revolution,
so here we consider layers which are invariant 
with respect to rotations around a fixed axis in~$\Real^3$.

Let $r,z : [0,\infty) \to \Real$ be smooth functions such that
$r(s) > 0$ for all $s>0$, 
$r(0) = 0 = z(0)$, $r'(0)=1$ and 
\begin{equation}\label{unit}
  r'(s)^2 + z'(s)^2 = 1
\end{equation}
for all $s \geq 0$. 
The last identity implies 
that the planar curve $\Gamma(s):=(r(s),z(s))$ is unit-speed
(see Figure~\ref{Fig.right}). 
We consider the smooth surface of revolution~$\Sigma$ 
obtained be revolving~$\Gamma$ along the second axis:
\begin{equation}\label{surface}
  \Sigma := \left\{
  \left(r(s) \cos\vartheta,r(s) \sin\vartheta,z(s)\right)
  : \ (s,\vartheta) \in [0,\infty) \times [0,2\pi)
  \right\}
  .
\end{equation}
We use the following natural parametrisation of~$\Sigma$:
\begin{equation}
  p: (0,\infty) \times (0,2\pi) \to \Real^3 :
  \left\{(s,\vartheta) \mapsto
  \left(r(s) \cos\vartheta,r(s) \sin\vartheta,z(s)\right)
  \right\}
  ,
\end{equation}
which gives rise to geodesic polar ``coordinates'' 
$(s,\vartheta)$ on~$\Sigma$. 
With respect to these coordinates, 
the induced metric $g_{\mu\nu}:=\partial_\mu p \cdot \partial_\nu p$,
where the dot denotes the scalar product in~$\Real^3$, 
reads 
\begin{equation}
  (g_{\mu\nu}) = 
  \begin{pmatrix}
    1 & 0 \\
    0 & r^2
  \end{pmatrix}  
  .
\end{equation}
In particular, the surface element of~$\Sigma$ reads 
$\der\Sigma := r(s) \, \der s \wedge \der\vartheta$.
Because of the availability of a unique chart~$p^{-1}$ 
(which covers the whole~$\Sigma$ except for the curve~$\Gamma$,
which is a set of measure zero relative to~$\Sigma$),
we shall consider the geometric objects of~$\Sigma$ 
as functions of $(s,\vartheta)$ rather than points of~$\Sigma$. 

With respect to the surface normal 
\begin{equation}
  n(s,\vartheta) :=
  \left(-z'(s) \cos\vartheta,
  -z'(s) \sin\vartheta,
  r'(s)\right)
  ,
\end{equation}
we have the following formulae
for the second fundamental form 
$h_{\mu\nu} := -\partial_\mu n \cdot \partial_\nu p$
and the Weingarten tensor $L^\mu_{\ \nu} = g^{\mu\rho} h_{\rho\nu}$:
\begin{equation}\label{principal}
\begin{aligned}
  (h_{\mu\nu}) &= 
  \begin{pmatrix}
    r' z'' - r'' z'
    & 0 \\
    0 & r z'
  \end{pmatrix}  
  , \\
  (L^\mu_{\ \nu}) &= 
  \begin{pmatrix}
    k_1 
    & 0 \\
    0 & k_2 
  \end{pmatrix}
  ,  
  \qquad
  \begin{aligned}
    k_1 &:= r' z'' - r'' z' \,,
    \quad
    k_2 := \frac{z'}{r} \,.
  \end{aligned}
\end{aligned}  
\end{equation}
Consequently, the layer metric~\eqref{metric.bis} is actually diagonal. 
The principal curvatures~$k_1$ and~$k_2$ will be called
the \emph{meridian} and \emph{parallel} curvatures, respectively. 
Because of the rotational symmetry, the curvatures are independent of~$\vartheta$,
so we suppress this variable from the arguments to simplify the notation.

Differentiating~\eqref{unit}, we obtain the identity $r'r''+z'z''=0$. 
Using it in the definition of $K=k_1k_2$ 
with help of the formulae~\eqref{principal} 
for the principal curvatures, 
we arrive at the Jacobi equation
\begin{equation}\label{Jacobi}
  r'' + K r = 0 
  ,
\end{equation}
subject to initial conditions $r(0)=0$ and $r'(0)=1$.
The differential equation~\eqref{Jacobi} 
has important consequences.
First, we have the following upper bound on the Jacobian~$r$.
\begin{Lemma}\label{Lem1} 
Assume~\eqref{Ass.K}. Then there exists a positive constant~$C_\Sigma$
such that  
$$
  \forall s \geq 0
  \,, \qquad
  r(s) \leq C_\Sigma \, s 
  \,.
$$ 
\end{Lemma}
\begin{proof}
Integrating~\eqref{Jacobi}, we obtain the uniform bound
$$
  r'(s) = 1 - \int_0^s (Kr)(s) \, \der s 
  \leq 1 + \frac{\|K\|_{L^1(\Sigma)} }{2\pi} 
  =: C_\Sigma
$$
for all $s \geq 0$.
Integrating this inequality, we obtain the desired claim.
\end{proof}

Second, \eqref{Jacobi} implies the Gauss--Bonnet theorem
\begin{equation}\label{Gauss-Bonnet}
  \mathcal{K} = 2\pi \, [1-r'(\infty)]
  \,, \qquad \mbox{where} \qquad
  r'(\infty) := \lim_{s\to\infty} r'(s)
  \,.
\end{equation}
Note that the limit is well defined as a consequence 
of this equality and hypothesis~\eqref{Ass.K}. 
Necessarily,
\begin{equation}\label{necessarily}
  0 \leq \mathcal{K} \leq 2\pi
  \,.
\end{equation}
Here the non-negativity follows from~\eqref{unit},
while the upper bound is valid due to the positivity of~$r$. 
If the total Gauss curvature~$\mathcal{K}$ is positive,
we get an important information on the parallel curvature.
\begin{Lemma}\label{Lem2} 
Assume~\eqref{Ass.K} with $\mathcal{K}  > 0$.
There exist positive numbers~$\delta$ and~$s_0$ such that
$$
  \forall s \geq s_0 \,, \qquad
  \frac{\delta}{r(s)} \leq |k_2(s)| \leq \frac{1}{r(s)}
  \,.
$$ 
\end{Lemma}
\begin{proof}
The lemma is due to~\cite[Lem.~6.1]{DEK2},
we repeat the proof to make the presentation self-contained. 
By~\eqref{necessarily}, \eqref{Gauss-Bonnet} 
and the assumption $\mathcal{K}  > 0$,
one has $0 \leq r'(\infty) < 1$.
It follows that there exist $\delta' \in (0,\frac{1}{2})$ and $s_0 > 0$
such that $-\delta' \leq r'(s) \leq 1-\delta'$ for all $s \geq s_0$.
Then the desired claim follows from the definition of~$k_2$.
\end{proof}

It follows from Lemmata~\ref{Lem1} and~\ref{Lem2}
that~$k_2$ is not integrable, \ie, $k_2 \not\in L^1((0,\infty))$. 
On the other hand, the meridian curvature~$k_1$ is integrable,
which follows from the smoothness of~$r,z$
and the following estimates:
\begin{equation}\label{integrable}
  \infty > \frac{\|K\|_{L^1(\Sigma)} }{2\pi} 
  \geq \int_{s_0}^\infty |k_1(s) k_2(s)| \, r(s) \, \der s 
  \geq \delta \int_{s_0}^\infty |k_1(s)| \, \der s 
  \,.
\end{equation}
Consequently, $M \not\in L^1((0,\infty))$. 
This is the essence of our subsequent analysis:
Even if~$M$ may decay at infinity,
it is not negligible in the integral sense there. 

Since~$\Sigma$ is rotationally symmetric,
the cut-radius maps $c_\pm$ do not depend 
on the angular variable, so we may suppress it from the argument. 
The map~$\Phi$ defined in~\eqref{layer}, 
when restricted to the open set
\begin{equation}\label{U-symmetric}
  U := \{ (s,\vartheta,t) \in (0,\infty) \times (0,2\pi) \times \Real 
  : -c_-(s) < t < c_+(s)
  \}
\end{equation}
is a diffeomorphism onto 
$
  \Phi(U) = \Real^3 \setminus 
  \{(x^1,0,x^3):x^1 \geq 0, x^3 \in \Real \}
$.
Since the latter coincides with~$\Real^3$ up to a set of measure zero,
the map~$\Phi$ can be used as a parametrisation of~$\Real^3$. 

The basic hypothesis~\eqref{Ass.local} is obviously
satisfied (with a positive~$a$) due to~\eqref{Ass.planar}
and smoothness of~$\Sigma$.  
The global requirement~\eqref{Ass.global}
must still be satisfied \emph{ad hoc} 
due to possible self-crossings of~$\Gamma$.
The following observation shows, however,
that it is actually enough
to exclude the self-crossings only locally.
\begin{Lemma}\label{Lem.less} 
Assume~\eqref{Ass.K} and~\eqref{Ass.planar}.
Then $r(s) \to \infty$ as $s\to\infty$.
\end{Lemma}
\begin{proof}
If~\eqref{Ass.K} holds with $\mathcal{K} > 0$,
then the result follows from Lemma~\ref{Lem2}
(note that $K,M$ vanish at infinity if, and only if,
$k_1,k_2$ vanish at infinity).
If $\mathcal{K} = 0$, then $r'(\infty) = 1$.
But then 
$
  r(s) = r(s_0) + \int_{s_0}^{s} r'
  \geq \frac{1}{2} (s-s_0)
$
for all sufficiently large $s_0 < s$.
Fixing~$s_0$ and sending~$s$ to~$\infty$,
we get the desired claim. 
\end{proof}

Our main hypothesis for rotationally symmetric layers
is that the cut-locus of~$\Sigma$ asymptotically coincides
with the upper part of the axis of symmetry:
\begin{equation}\label{Ass.cut}
  \fbox{$
  \exists R_0>0 \,, \qquad
  \Cut(\Sigma) \setminus B_{R_0}(0)
  = \big\{(0,0,x^3): x^3 \geq R_0\big\}
  $.} 
\end{equation}
\begin{Lemma}\label{Lem3} 
Assume~\eqref{Ass.K} with $\mathcal{K}  > 0$, \eqref{Ass.planar}
and~\eqref{Ass.cut}.
Then there exists $s_0 > 0$ such that,
for all $s \geq s_0$, $k_1(s) \geq 0$,
$k_2(s) > 0$, $c_-(s) = \infty$, 
$c_+(s) = 1/k_2(s)$
and $k_1(s) \leq k_2(s)$.
\end{Lemma}
\begin{proof}
The crucial observation is that $r(s) \to \infty$ as $s \to \infty$
as a consequence of Lemma~\ref{Lem2} and~\eqref{Ass.planar}
(or see directly Lemma~\ref{Lem.less}).
At the same time, $z(s) \to \pm\infty$ as $s \to \infty$
because $z'(\infty)^2 = 1-r'(\infty)^2 > 0$,
where the inequality is implied by $\mathcal{K}  > 0$ 
and~\eqref{Gauss-Bonnet}. 
Then $c_-(s) = \infty$ and $c_+(s) = 1/k_2(s)$ for all sufficiently large~$s$.
Indeed, by~\eqref{Ass.cut} and the symmetry,
for all sufficiently large~$s$ 
and any $\vartheta \in [0,2\pi)$, the curve 
$\gamma(t) := p(s,\vartheta)+n(s,\vartheta) t$  
does not intersect the ball $B_{R_0}(0)$,
so the only intersection must be with the semi-axis 
$\big\{(0,0,x^3): x^3 \geq R_0\big\}$.
Moreover, $\gamma^1(t) = 0$ implies $t = r/z' = 1/k_2$.
This argument also excludes the possibility 
$z(s) \to -\infty$ as $s \to \infty$,
because otherwise the curve~$\gamma$ 
would intersect the negative semi-axis $\{(0,0,x^3): x^3 \leq -R_0\}$
for all sufficiently large~$s$.
Consequently, $z'(\infty) > 0$, 
so the parallel curvature $k_2(s)$ is positive 
for all sufficiently large~$s$.
We claim that the meridian curvature $k_1(s)$ is non-negative 
for all sufficiently large~$s$
(see Figure~\ref{Fig.right}).
Indeed, if it is not the case, 
then there exist large positive numbers $s_1<s_2$ 
such that the graph of the curve~$\Gamma$ is strictly concave on $(s_1,s_2)$,
implying the existence of a cut-locus of~$\Gamma$ 
to the right of the curve
(when traced according to the arc-length parameter), 
therefore violating~\eqref{Ass.cut}.
Finally, if $k_1(s) > k_2(s)$, then $1/k_1(s) < 1/k_2(s) = c_+(s)$,
implying a contradiction that there exists a conjugate point 
outside the cut-locus.
\end{proof}

There are many surfaces of revolution satisfying~\eqref{Ass.cut}.
For instance, if $s \mapsto \Gamma(s)$ is a convex graph,
then $c_-(s) = \infty$ for all $s \geq 0$,
so there is no cut-locus ``outside'' of~$\Sigma$ (\ie~for negative~$t$).
To show that the cut-locus satisfies~\eqref{Ass.cut} 
``inside'' of~$\Sigma$ (\ie~for positive~$t$),
one can employ the geometric interpretation 
of the principal curvatures (minimal and maximal values
of the normal curvatures of all the curves passing through a given point).
Then it is easy to verify that~\eqref{Ass.cut}
can be achieved as a consequence~\eqref{Ass.planar} 
(still assuming \eqref{Ass.K} with $\mathcal{K}  > 0$).
In particular, the paraboloid of revolution satisfies~\eqref{Ass.cut}
(as well as \eqref{Ass.K} with $\mathcal{K}  = 2\pi$ and~\eqref{Ass.planar}).

To be even more explicit, 
let us consider the family of surfaces~$\Sigma_\theta$ 
obtained be revolving the planar curve~\eqref{curve.special}.
The cut-locus of~$\Sigma_\theta$ is the semi-axis $\{(0,0,x^3) : x^3 \geq R \}$
if $\theta \in (0,\pi]$, while it is empty if $\theta=0$.
Of course, $\Sigma_\theta$ is not smooth (unless $\theta=0$),
but it is piecewise smooth (in fact, piecewise analytic).
The total Gauss curvature of~$\Sigma_\theta$ reads
\begin{equation}\label{total}
  \mathcal{K}_\theta = 2\pi (1-\cos\mbox{$\frac{\theta}{2}$})
  \,,
\end{equation}
so the hypotheses of Lemma~\ref{Lem2} are met whenever $\theta \in (0,\pi)$. 
     
Finally, we make a hypothesis about a sufficient decay 
of the parallel curvature at infinity:   
\begin{equation}\label{Ass.decay}
  \fbox{$
  \exists \epsilon>0 \,, \qquad
  k_2(s) = O(s^{-\epsilon})
  \qquad\mbox{as}\qquad
  s \to \infty
  $.} 
\end{equation}
This condition (with $\epsilon=1$) 
is easily verified for~$\Sigma_\theta$      
whenever $\theta \in [0,\pi)$. 
It also holds for the paraboloid of revolution (with $\epsilon=1/2$)
and other surfaces obtained by revolving polynomially growing curves.

\begin{figure}[h]
\begin{center}
\begin{tabular}{c}
\includegraphics[width=0.8\textwidth]{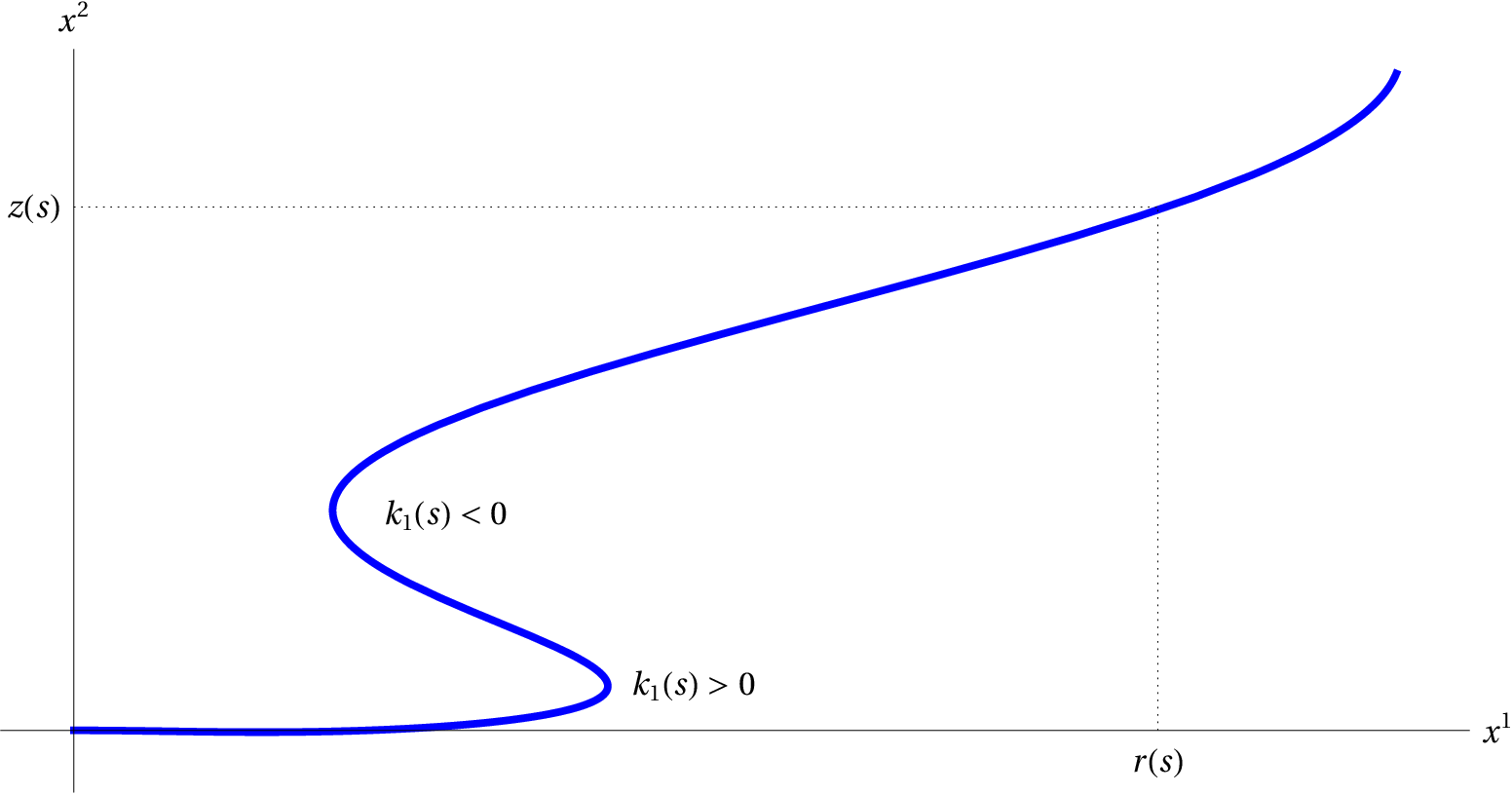}
\end{tabular}
\end{center}
\caption{The geometry of the generating curve~$\Gamma$.} 
\label{Fig.right}
\end{figure}
% 

%------------------------------------------------------% 
\section{The proofs}\label{Sec.proof}
%------------------------------------------------------% 
% 
We assume that the potential~$V$ is either the distribution
$V_\mathrm{leaky} := \alpha \delta_\Sigma$ with $\alpha<0$
or it is an essentially bounded function
satisfying~\eqref{Ass.vary} and~\eqref{Ass.attractive}.
Let $E_1<0$ denote the lowest discrete eigenvalue of~$T$.
The variational characterisation yields
\begin{equation}\label{Rayleigh}
  E_1 = \inf_{\stackrel[\xi \not= 0]{}{\xi \in H^1(\Real)}} 
  \frac{\displaystyle \int_\Real |\xi'(t)|^2 \, \der t
  +  \int_\Real W(t) \, |\xi(t)|^2 \, \der t}
  {\displaystyle \int_\Real |\xi(t)|^2 \, \der t} 
  \,.
\end{equation}
In the leaky case, 
the integral
$
  \int_\Real W(t) \, |\xi(t)|^2 \, \der t 
$
should be interpreted as
$  
  \alpha \, |\xi(0)|^2
$,
in which case, explicitly, 
$E_1 = -\frac{\alpha^2}{4}$.
It is well known that~$E_1$ is simple
and that the corresponding eigenfunction~$\xi_1$
can be chosen to be positive.
We additionally choose the eigenfunction
to be normalised to~$1$ in $\sii(\Real)$,
\ie, $\|\xi_1\|_{\sii(\Real)}=1$.
Explicitly,
$
  \xi_1(t) = \sqrt{\frac{|\alpha|}{4}} e^{\frac{\alpha}{2} |t|}
$
in the leaky case.
In any case, one knows that
$\xi_1 \in H^1(\Real) \cap L^\infty(\Real)$
and that the following identities hold true:
\begin{equation}\label{asymptotics}
  \xi_1(t)
  = N_\pm e^{\mp\sqrt{-E_1} t}
  \qquad \mbox{for every} \qquad
  \pm t >  a
  \,,
\end{equation}
where~$N_\pm$ are positive constants.

\begin{Remark}
In principle, the assumption $\supp W \subset [-a,a]$
of~\eqref{Ass.vary} could be relaxed to a decay of~$W$ at infinity. 
Then the asymptotics~\eqref{asymptotics} could 
be replaced by Agmon-type estimates.
\end{Remark}

First of all, we locate the essential spectrum of~$H$
(assuming $\supp W \subset [-a,a]$ or the leaky setting).
As an auxiliary quantity,
in parallel with~\eqref{Rayleigh},
we consider  
\begin{equation}\label{Rayleigh.cut}
  E_1^b 
  := \inf_{\stackrel[\xi \not= 0]{}{\xi \in H^1((-b,b))}} 
  \frac{\displaystyle \int_{-b}^{b} |\xi'(t)|^2 \, \der t
  +  \int_{-b}^{b} W(t) \, |\xi(t)|^2 \, \der t}
  {\displaystyle \int_{-b}^{b} |\xi(t)|^2 \, \der t} 
  \,,
\end{equation}
where $b > a$.
Of course, $E_1^b$ is the lowest eigenvalue
of the operator~$T$ restricted to $(-b,b)$,
subject to Neumann boundary conditions.
\begin{Lemma}\label{Lem.gap}
One has
$$
  \lim_{b \to \infty} E_1^b = E_1 \,.
$$  
\end{Lemma}
\begin{proof}
In the leaky case, $E_1^b$ solves the implicit equation
$2 \sqrt{-E} = -\alpha \coth(\sqrt{-E} b)$,
from which the convergence can be easily deduced.   
In the regular case, let us assume $b > a$.
By using~$\xi_1$ 
(or, more precisely, its restriction to $(-b,b)$)
as a trial function in~\eqref{Rayleigh.cut},
it is easy to see that 
\begin{equation}\label{easy}
  E_1^b \leq E_1
  \,.
\end{equation}
To get an opposite estimate, 
let~$\xi_1^b$ be the positive minimiser of~\eqref{Rayleigh.cut}
normalised to~$1$ in $\sii((-b,b))$. 
We extend it to the whole line by setting
$$
  \tilde{\xi}_1^b(t) := 
  \begin{cases}
    \xi_1^b(t) 
    & \mbox{if} \quad |t| < b \,, 
    \\
    \xi_1^b(\pm b) \, \exp\big({\mp \sqrt{-E_1} (t \mp b)}\big) 
    & \mbox{if} \quad \pm t \geq b \,.
  \end{cases}
$$
Since $\tilde{\xi}_1^b \in H^1(\Real)$,
we use it as a trial function in~\eqref{Rayleigh} and obtain
$$
  E_1 \leq \frac{E_1^b + \frac{1}{2} \sqrt{-E_1} \,
  [\xi_1^b(-b)^2 + \xi_1^b(b)^2]}
  {1+ \frac{1}{2} \frac{1}{\sqrt{-E_1}} \,
  [\xi_1^b(-b)^2 + \xi_1^b(b)^2]}
  \,.
$$
It remains to notice that $\xi_1^b(\pm b) \to 0$ as $b \to \infty$.
To see it, 
we employ the explicit solution 
$$
  \xi_1^b(t) = \xi_1^b(\pm b) \,
  \cosh\big( \mbox{$\sqrt{-E_1^b}$} \, (t \mp b) \big)
  \qquad \mbox{for} \qquad 
  a \leq \pm t \leq b 
  \,.
$$
Take $t = \pm a$ 
and use that the value $\xi_1^b(\pm a)$ can be estimated
by the $H^1((-a,a))$ norm of~$\xi_1^b$ as follows:
$$
  \xi_1^b(\pm a)^2 \leq \int_{-a}^a  |{\xi_1^b}'(t)|^2 \, \der t
  + C_a \int_{-a}^a  |{\xi_1^b}(t)|^2 \, \der t
  \,,
$$ 
where explicitly $C_a := 1 + (2a)^{-1}$.
In turn, 
the right-hand side 
can be estimated by using the identity
\begin{equation*}%\label{weak}
   \int_{-b}^{b} |{\xi_1^b}'(t)|^2 \, \der t
  +  \int_{-b}^{b} W(t) \, |\xi_1^b(t)|^2  \, \der t
  = E_1^b \int_{-b}^{b} |\xi_1^b(t)|^2  \, \der t \,.
\end{equation*}
Consequently, 
$
  \xi_1^b(\pm a)^2 \leq \|W\|_\infty + C_a 
$.
Finally, we obtain the exponential decay
$$
  \xi_1^b(\pm b) \leq 
  \frac{\sqrt{\|W\|_\infty + C_a} }
  {\cosh\big(\mbox{$\sqrt{-E_1^b}$} \, (a \mp b) \big)}
  \xrightarrow[b \to \infty]{}
  0
  \,.
$$
This concludes the proof of the lemma.
\end{proof}

The following theorem does not require that~$\Sigma$ 
is a surface of revolution.
\begin{Theorem}\label{Thm.ess.general} 
Let~$\Sigma$ be an orientable smooth surface
which is asymptotically cut-locus
planar~\eqref{Ass.planar.cut}. 
Let the tubular neighbourhood~\eqref{tube}
does not overlap itself with some positive~$a$,
\ie~\eqref{Ass.local} and~\eqref{Ass.global} hold. 
Let~$V$ satisfy~\eqref{Ass.vary} and~\eqref{Ass.attractive}. 
Then 
$$
  \inf\sigma_\mathrm{ess}(H) \geq E_1
  \,.
$$
\end{Theorem}
\begin{proof}  
Fixing a point $p_0 \in \Sigma$ and giving any positive number~$R$,
we divide the surface~$\Sigma$ into two parts 
$\Sigma_\mathrm{int} := \Sigma \cap B_R(p_0)$
and $\Sigma_\mathrm{ext} := \Sigma \setminus \overline{B_R(p_0)}$,
where $B_R(p_0)$ is the geodesic ball of radius~$R$ centred at~$p_0$. 
Correspondingly, we divide the set~$U$ into two parts
$$
\begin{aligned}
  U_\mathrm{int} 
  &:= \{ (p,t) \in  \Sigma_\mathrm{int} \times \Real
  : \ -c_-(p) < t < c_+(p)
  \} \,,
  \\
  U_\mathrm{ext} 
  &:= \{ (p,t) \in  \Sigma_\mathrm{ext} \times \Real
  : \ -c_-(p) < t < c_+(p)
  \} \,.
\end{aligned}
$$  
The interior part is further subdivided into two subparts
$$
\begin{aligned}
  U_\mathrm{int,1} 
  &:= \{ (p,t) \in  \Sigma_\mathrm{int} \times \Real
  : \ -a < t < a
  \} \,,
  \\
  U_\mathrm{int,2}
  &:= \{ (p,t) \in  \Sigma_\mathrm{int} \times \Real
  : \ -c_-(p) < t < -a \ \lor \ a < t < c_+(p)  
  \} \,.
\end{aligned}
$$ 
Analogously, we subdivide $U_\mathrm{ext}$
into two subparts 
$$
\begin{aligned}
  U_\mathrm{ext,1} 
  &:= \{ (p,t) \in  \Sigma_\mathrm{ext} \times \Real
  : \ -b < t < b
  \} \,,
  \\
  U_\mathrm{ext,2}
  &:= \{ (p,t) \in  \Sigma_\mathrm{ext} \times \Real
  : \ -c_-(p) < t < -b \ \lor \ b < t < c_+(p)  
  \} \,,
\end{aligned}
$$ 
where $b>0$.
By~\eqref{Ass.planar.cut},
we can assume $b > a$ by choosing~$R$ large enough.  
Define
$$
  c_R^\pm := 1 \pm 2 \, \|M\|_R \, b \pm \|K\|_R \, b^2
  \qquad \mbox{with} \qquad
  \|\cdot\|_R := \|\cdot\|_{L^\infty(\Sigma_\mathrm{ext})}
  \,.
$$
Since~\eqref{Ass.planar.cut} implies~\eqref{Ass.planar},
given any (large) $b > a$, there exists (large)~$R$ 
such that~$c_R^-$ is positive. 

We consider the auxiliar operator~$\hat{H}^N$ which is obtained from~$\hat{H}$
by impossing an extra Neumann condition 
(\ie~no condition on the level of sesquilinear forms)
on the boundaries of the subsets described above. 
More specifically, 
$
  \hat{H}^N 
  = \hat{H}_\mathrm{int,1}^N 
  \oplus \hat{H}_\mathrm{int,2}^N 
  \oplus \hat{H}_\mathrm{ext}^N
$,
where 
$
  \hat{H}_\mathrm{ext}^N 
  = \hat{H}_\mathrm{ext,1}^N 
  \oplus \hat{H}_\mathrm{ext,2}^N 
$ 
is the self-adjoint operator associated
with the form $\hat{h}_\mathrm{ext}^N$ in $\sii(U_\mathrm{ext},\der v)$
defined by 
$$
\begin{aligned}
  \hat{h}_\mathrm{ext}^N[\psi] 
  &:= \int_{U_\mathrm{ext}} 
  \overline{\partial_\mu\psi} \,  G^{\mu\nu} \, \partial_\nu\psi
  \, \der v
  + \int_{U_\mathrm{ext}} |\partial_t\psi|^2  \, \der v
  + \int_{U_\mathrm{ext}} W(t) \, |\psi|^2 \, \der v \,,
  \\
  \dom \hat{h}_\mathrm{ext}^N[\psi] 
  &:= \big\{\psi \upharpoonright \hat{h}_\mathrm{ext} : \
  \psi \in \dom \hat{h}
  \big\} \,,
\end{aligned} 
$$
and similarly for the other operators.
Obviously, $\dom \hat{h}^N \supset \dom \hat{h}$,
therefore $\hat{H}^N \leq \hat{H}$ in the sense of quadratic forms,
so, by the minimax principle, it is enough to show that  
$\inf\sigma_\mathrm{ess}(\hat{H}^N) \geq E_1$.
Since $\supp W \subset [-a,a]$ due to~\eqref{Ass.vary},
the operator $\hat{H}_\mathrm{int,2}^N$ is non-negative,
so the inequality  
$
  \inf\sigma_\mathrm{ess}(\hat{H}_\mathrm{int,2}^N) 
  %\inf\sigma(\hat{H}_\mathrm{int,2}^N) 
  %\geq 0 
  \geq E_1
$
is trivial.
At the same time, $\hat{H}_\mathrm{int,1}^N$ is an operator
with compact resolvent, so it does not contribute to the essential
spectrum of $\hat{H}^N$ 
(one has $\inf\sigma_\mathrm{ess}(\hat{H}_\mathrm{int,1}^N) = \infty$ 
by the minimax principle).
It remains to estimate the essential spectrum of $\hat{H}_\mathrm{ext}^N$.

For every $\psi \in \dom \hat{h}_\mathrm{ext}^N$,
$$
\begin{aligned}
  \hat{h}_\mathrm{ext}^N[\psi] 
  &\geq \int_{U_\mathrm{ext}} |\partial_t\psi|^2  \, \der v
  + \int_{U_\mathrm{ext}} W(t) \, |\psi|^2 \, \der v 
  \\
  &\geq \int_{U_\mathrm{ext,1}} |\partial_t\psi|^2  \, \der v
  + \int_{U_\mathrm{ext,1}} W(t) \, |\psi|^2 \, \der v 
  \\
  &\geq c_R^-
  \int_{U_\mathrm{ext,1}} |\partial_t\psi|^2  \, \der \Sigma \, \der t
  + \int_{U_\mathrm{ext,1}} W(t) \, |\psi|^2 \, \der v 
  \\
  &\geq c_R^- \, E_1^b
  \int_{U_\mathrm{ext,1}} |\psi|^2  \, \der \Sigma \, \der t
  - c_R^- \int_{U_\mathrm{ext,1}} W(t) \, |\psi|^2 \, \der \Sigma \, \der t
  + \int_{U_\mathrm{ext,1}} W(t) \, |\psi|^2 \, \der v 
  \\
  &\geq \frac{c_R^-}{c_R^+} \, E_1^b \,
  \|\psi\|_{\sii(U_\mathrm{ext},\der v)}^2
  - c_R^- \int_{U_\mathrm{ext,1}} W(t) \, |\psi|^2 \, \der \Sigma \, \der t
  + \int_{U_\mathrm{ext,1}} W(t) \, |\psi|^2 \, \der v 
\end{aligned} 
$$
where the fourth inequality employs the variational
definition of~$E_1^b$ (\cf~\eqref{Rayleigh.cut})
with help of Fubini's theorem
and the fact that $\supp W \subset [-a,a]$.
Since 
$$
  \int_{U_\mathrm{ext,1}} W(t) \, |\psi|^2 \, \der v 
  \geq \int_{U_\mathrm{ext,1}} c_R(t) \, W(t) \, |\psi|^2 \, \der \Sigma \, \der t
  \,,
$$
where $c_R(t) := c_R^-$ if $W(t) \geq 0$ 
and $c_R(t) := c_R^+$ if $W(t) < 0$, 
we get 
$$
\begin{aligned}
  \hat{h}_\mathrm{ext}^N[\psi] 
  &\geq \frac{c_R^-}{c_R^+} \, E_1^b \,
  \|\psi\|_{\sii(U_\mathrm{ext},\der v)}^2
  - (c_R^+ - c_R^-) \int_{U_\mathrm{ext,1}} 
  \chi_{\{W(t) < 0\}}(t) \, 
  |W(t)| \, |\psi|^2 \, \der \Sigma \, \der t
  \\
  &\geq 
  \left(
  \frac{c_R^-}{c_R^+} \, E_1^b
  - \frac{c_R^+ - c_R^-}{c_R^-}
  \, \|W\|_\infty
  \right)
  \|\psi\|_{\sii(U_\mathrm{ext},\der v)}^2
  \,.
\end{aligned} 
$$
In summary,
$$
  \inf\sigma_\mathrm{ess}(\hat{H}) 
  \geq \inf\sigma_\mathrm{ess}(\hat{H}_\mathrm{ext}^N) 
  \geq \inf\sigma(\hat{H}_\mathrm{ext}^N) 
  \geq \frac{c_R^-}{c_R^+} \, E_1^b
  - \frac{c_R^+ - c_R^-}{c_R^-}
  \, \|W\|_\infty
  \,.
$$
Since $c_R^-/c_R^+ \to 1$ as $R \to \infty$ due to~\eqref{Ass.planar}
(which is a consequence of~\eqref{Ass.planar.cut}),
we obtain
$\inf\sigma_\mathrm{ess}(\hat{H}) \geq E_1^b$.
Finally, the arbitrariness of~$b$ and Lemma~\ref{Lem.gap} 
yield that $\inf\sigma_\mathrm{ess}(\hat{H}) \geq E_1$.
\end{proof}

We leave as an open problem whether 
$\sigma_\mathrm{ess}(H) \supset [E_1,\infty)$
under the hypotheses of Theorem~\ref{Thm.ess.general}.
 
Now we turn to the existence of bound states. 
We heavily rely on results in Section~\ref{Sec.symmetry}
for rotationally symmetric layers.
In particular, recall that,
under the hypotheses of the following theorem: 
the surface Jacobian~$r(s)$
is bounded from above by a multiple of~$s$ (Lemma~\ref{Lem1})
and diverges as $s \to \infty$ (Lemma~\ref{Lem.less});
the parallel curvature~$k_2$ behaves like $r^{-1}$ (Lemma~\ref{Lem2});
the cut-radius maps satisfy $c_-(s) = \infty$ and $c_+(s) = k_2(s)^{-1}$
for all sufficiently large~$s$ (Lemma~\ref{Lem3});
and thus \eqref{Ass.planar.cut} follows 
as a consequence of~\eqref{Ass.planar}.
\begin{Theorem}\label{Thm.main.bis} 
Let~$\Sigma$ be a surface of revolution given by~\eqref{surface}
and satisfying~\eqref{Ass.K} with $\mathcal{K} > 0$ 
and~\eqref{Ass.planar}.
Let~$V$ satisfy~\eqref{Ass.vary} with some positive~$a$ 
and~\eqref{Ass.attractive}.
Assume in addition~\eqref{Ass.cut} and~\eqref{Ass.decay}. 
Then~$H$ possesses an infinite number 
(counting multiplicities) of discrete eigenvalues below~$E_1$.
\end{Theorem}
\begin{proof} 
In view of Theorem~\ref{Thm.ess.general}, 
to establish the existence 
of a discrete eigenvalue of~$H$,
it is enough to show that $\inf\sigma(H) < E_1$.
By the minimax principle, it is thus enough to find
a trial function $\psi \in \dom Q := \dom \hat{h}$ 
such that $Q[\psi] := \hat{h}[\psi] - E_1 \, \|\psi\|^2 < 0$,
where~$\|\cdot\|$ denotes the norm in $\sii(U,\der v)$.
Recall that, in the rotationally symmetric case, we have
\begin{equation}\label{reader}
\begin{aligned}
  \hat{h}[\psi] &= 
  \int_U |\partial_s\psi(s,\vartheta,t)|^2 \, 
  \frac{1-k_2(s) t}{1-k_1(s) t} 
  \, \der \Sigma \, \der t
  +  \int_U \frac{|\partial_\vartheta\psi(s,\vartheta,t)|^2}{r(s)^2} \, 
  \frac{1-k_1(s) t}{1-k_2(s) t} 
  \, \der \Sigma \, \der t
  \\
  & \quad
  + \int_U |\partial_t\psi(s,\vartheta,t)|^2 \,
  (1-k_1(s) t)(1-k_2(s) t)
  \, \der \Sigma \, \der t
  \\
  \|\psi\|^2 &= 
  \int_U |\psi(s,\vartheta,t)|^2 \,
  (1-k_1(s) t)(1-k_2(s) t)
  \, \der \Sigma \, \der t
  \,,
\end{aligned}
\end{equation}
where $\der\Sigma = r(s) \, \der s \wedge \der \vartheta$
and~$U$ is given by~\eqref{U-symmetric}.

For every real~$\eps$, we introduce
a $\vartheta$-independent trial function
\begin{equation*}
  \psi_{n,\eps}(s,\vartheta,t) := \varphi_n(s) \, \xi_1(t)
  + \eps \, \phi_n(s) \, t \, \xi_1(t)
  \,,
\end{equation*}
where the sequence $\{\varphi_n\}_{n=2}^\infty$ is defined by
\begin{equation*}
  \varphi_n(s) := 
  \begin{cases}
    0 & \mbox{if} \quad s \in [0,n) \,,
    \\
    \displaystyle
    \frac{\log(s/n)}{\log n} 
    & \mbox{if} \quad s \in [n,n^2) \,,
    \\
    \displaystyle
    \frac{\log(n^3/s)}{\log n} 
    & \mbox{if} \quad s \in [n^2,n^3) \,,
    \\
    0 & \mbox{if} \quad s \in [n^3,\infty) \,,
  \end{cases}
  \qquad \mbox{and} \qquad
  \phi_n(s) := \frac{\varphi_n(s)}{s}
  \,.
\end{equation*}
Note that the supports of~$\varphi_n$ and~$\phi_n$
tend to infinity as $n \to \infty$.
We always assume that~$n$ is so large that
the asymptotic properties of Lemma~\ref{Lem3} hold.
Proceeding as in~\cite[Lem.~1]{KKK} (see also below),
one can verify that $\psi_{n,\eps} \in \dom Q$. 
Then
\begin{equation}
  Q[\psi_{n,\eps}]
  = Q[\varphi_n \xi_1]
  + 2 \eps \, Q\big(\phi_n t \xi_1,\varphi_n \xi_1\big)
  + \eps^2 \, Q[\phi_n t \xi_1]
  \,.
\end{equation}
We make the decomposition $Q= Q_1 + Q_2$,
where 
%
%\begin{equation*} 
$
  Q_1[\psi]
  := \int_U \overline{\partial_{\mu}\psi} \, G^{\mu\nu} \, \partial_{\nu}\psi
  \, \der v  
$  
%\end{equation*}
%
is just the first line of~\eqref{reader}.

\noindent 
\fbox{$Q_1$} 
One has 
$$
  Q_1[\varphi_n\xi_1] 
  = \int_U \varphi_n'(s)^2 \, \xi_1(t)^2 
  \, f(s,t)
  \, \der \Sigma \, \der t
  \qquad \mbox{with} \qquad
  f(s,t) := \frac{1-k_2(s) t}{1-k_1(s) t} \geq 0
  \,.
$$
Since 
$$
  \partial_t f(\cdot,t) = \frac{k_1-k_2}{(1-k_1 t)^2} \leq 0
  \,,
$$
where the inequality holds due to Lemma~\ref{Lem3},
one has $\|f\|_\infty = 1$.
Consequently,
\begin{equation*}%\label{Q1.bound}
\begin{aligned}
  Q_1[\varphi_n\xi_1] 
  &\leq
  \int_U \varphi_n'(s)^2 \, \xi_1(t)^2 
  \, \der \Sigma \, \der t
  \\
  &\leq 
  2\pi C_\Sigma  
  \int_0^\infty \varphi_n'(s)^2 \, s \, \der s
  \\
  &= \frac{4\pi C_\Sigma}{\log n}
  \,,
\end{aligned}
\end{equation*}
where the second inequality holds due to Lemma~\ref{Lem1}
and the normalisation of~$\xi_1$.
Similarly,
\begin{equation*}
  Q_1[\phi_n t \xi_1]
  \leq \frac{4\pi C C_\Sigma}{n^2 \log n} 
  \,,
  \qquad \mbox{where} \qquad
  C := \int_{\Real} \xi_1(t)^2 \, t^2 \, \der t
  \,,
\end{equation*}
where the extra decay~$n^{-2}$ comes from the bound $s \geq n$
on the support of~$\phi_n$. 
The mixed term $Q_1\big(\phi_n t \xi_1,\varphi_n \xi_1\big)$
tend to zero as $n \to \infty$ by using these estimates
and the Schwarz inequality.
In summary,
\begin{equation*}
  \lim_{n \to \infty} Q_1[\psi_{n,\eps}] = 0
  \,.
\end{equation*}

\noindent  
\fbox{$Q_2$, order $\eps^0$} 
One has
\begin{equation}\label{Q2.bound} 
\begin{aligned}
  Q_2[\varphi_n \xi_1]
  &= \int_U \varphi_n(s)^2 \, \xi_1'(t)^2 \, 
  (1-2M(s) t + K(s)t^2) \, \der \Sigma \, \der t
  \\
  & \quad + \int_U W(t) \, \varphi_n(s)^2 \, \xi_1(t)^2 \, 
  (1-2 M(s) t + K(s) t^2) \, \der \Sigma \, \der t
  \\
  & \quad 
  - E_1 \int_U  \varphi_n(s)^2 \, \xi_1(t)^2 \, 
  (1-2 M(s) t + K(s) t^2) \, \der \Sigma \, \der t
  \\
  &= \int_U |\varphi_n(s)|^2 \xi_1(t)\xi_1'(t) \, 
  (2M(s) - 2K(s)t) \, \der \Sigma \, \der t
  \\
  & \quad + \int_{\Sigma} |\varphi_n(s)|^2
  \left[
  \xi_1(t)\xi'_1(t) (1-2 M(s) t + K(s) t^2) 
  \right]_{t=-c_-(s)}^{t=c_+(s)} \, \der \Sigma
  \\
  &= \int_U |\varphi_n(s)|^2 \, |\xi_1(t)|^2 \,  
  K(s) \, \der \Sigma \, \der t
  \\
  & \quad + \int_{\Sigma} |\varphi_n(s)|^2
  \left[
  \xi_1(t)^2 (M(s)-K(s)t)
  + \xi_1(t)\xi_1'(t) (1-2 M(s) t + K(s) t^2) 
  \right]_{t=-c_-(s)}^{t=c_+(s)} \, \der \Sigma
  \,.
\end{aligned}  
\end{equation}
Here the first equality follows by an integration by parts
and the identity $-\xi_1''+W\xi_1=E_1\xi_1$.
The second equality is a result of yet another 
integration by parts after writing $2\xi_1\xi_1' = (\xi_1^2)'$.
Recall that the support of~$\varphi_n$
tends to infinity as $n \to \infty$.
Then the resulting integral over~$U$ vanishes as $n \to \infty$ 
due to~\eqref{Ass.K} and the dominated convergence theorem.  
What is more, the resulting integral over~$\Sigma$  
vanishes as $n \to \infty$,
for the exponential decay of~$\xi_1$ dominates 
all the other functions that appear there.
More specifically, evaluating at $-c_-(s) = -\infty$ 
(recall Lemma~\ref{Lem3}) does not contribute.
Recalling that $c_+ = 1/k_2$, one has 
$(1-2M c_+ +K c_+^2) = (1-k_1 c_+)(1-k_2 c_+) = 0$
and $M-K c_+ = \frac{1}{2} (k_2 - k_1) \leq 1$
(because the curvatures vanish at infinity),
so it is enough to estimate (recall~\eqref{asymptotics})
$$
\begin{aligned}
  \left|
  \int_{\Sigma} |\varphi_n(s)|^2 \, \xi_1(c_+(s))^2 \, \der \Sigma
  \right|
  &\leq 2\pi \int_{n}^{n^3} \xi_1(c_+(s))^2 \, r(s) \, \der s
  \\
  &\leq 2\pi \, N_+^2  
  \int_{n}^{n^3} e^{-2\sqrt{-E_1}c_+(s)} \, c_+(s) \, \der s
  \\
  &\leq 2\pi \, N_+^2  \, (n^3-n) \,
  e^{-2\sqrt{-E_1}c_+(n)} \, c_+(n^3)  
  \,,
\end{aligned}  
$$
where the first inequality employs $r \leq c_+$.
The upper bound vanishes as $n \to \infty$ due to~\eqref{Ass.decay}.
In summary,
\begin{equation*}
  \lim_{n \to \infty} Q_2[\varphi_n \xi_1] = 0
  \,.
\end{equation*}

\noindent  
\fbox{$Q_2$, order $\eps^1$} 
Proceeding as in~\eqref{Q2.bound}, we have 
\begin{equation*} 
\begin{aligned}
  Q_2\big(\phi_n t \xi_1,\varphi_n \xi_1\big)
  &= \int_U \phi_n(s) \varphi_n(s) \,
  (t\xi_1(t))' \xi_1'(t) \, (1-2M(s) t + K(s)t^2) \, \der \Sigma \, \der t
  \\
  & \quad + \int_U W(t) \, 
  \phi_n(s) \varphi_n(s) \,
  t\xi_1(t) \, \xi_1(t) \,
  (1-2 M(s) t + K(s) t^2) \, \der \Sigma \, \der t
  \\
  & \quad
  - E_1 \int_U  \phi_n(s) \varphi_n(s) \, t\xi_1(t) \, \xi_1(t) \,
  (1-2 M(s) t + K(s) t^2) \, \der \Sigma \, \der t
  \\
  &= \int_U  \phi_n(s) \varphi_n(s) \, t \xi_1(t)\xi_1'(t) 
  (2M(s) - 2K(s)t) \, \der \Sigma \, \der t
  \\
  & \quad + \int_{\Sigma} \phi_n(s) \varphi_n(s) 
  \left[
  t\xi_1(t)\xi_1'(t) (1-2 M(s) t + K(s) t^2) 
  \right]_{t=-c_-(s)}^{t=c_+(s)} \, \der \Sigma
  \\
  &= -\int_U \phi_n(s) \varphi_n(s)  \, \xi_1(t)^2 \,  
  (M(s)-2 K(s) t) \, \der \Sigma \, \der t
  \\
  & \quad + \int_{\Sigma} \phi_n(s) \varphi_n(s)  
  \left[
  t\xi_1(t)^2 (M(s)-K(s)t)
  + t\xi_1(t)\xi_1'(t) (1-2 M(s) t + K(s) t^2) 
  \right]_{t=-c_-(s)}^{t=c_+(s)} \, \der \Sigma
  \,.
\end{aligned}  
\end{equation*}
Again, the term containing~$K$ in the resulting integral over~$U$
and the resulting integral over~$\Sigma$ vanish as $n \to \infty$. 
Consequently, as $n\to \infty$,
\begin{equation*} 
\begin{aligned}
  %\lim_{n\to\infty}
  Q_2\big(\phi_n t \xi_1,\varphi_n \xi_1\big)
  &= - %\lim_{n\to\infty}
  \int_U \phi_n(s) \varphi_n(s)  \, \xi_1(t)^2 \,  
  M(s) \, \der \Sigma \, \der t
  + o(1)
  \\
  &= - \frac{1}{2} %\lim_{n\to\infty}
  \int_U \phi_n(s) \varphi_n(s)  \, \xi_1(t)^2 \,  
  k_2(s) \, \der \Sigma \, \der t
  + o(1)
  \,,
\end{aligned} 
\end{equation*}
where the second equality follows from~\eqref{integrable}
and the dominated convergence theorem.
Here, employing Lemma~\ref{Lem2},  
\begin{equation*} 
\begin{aligned}
  \int_U \phi_n(s) \varphi_n(s)  \, \xi_1(t)^2 \,  
  k_2(s) \, \der \Sigma \, \der t
  &\geq 2\pi \delta
  \int_0^\infty \frac{\varphi_n(s)^2}{s}
  \int_{-c_-(s)}^{c_+(s)} \xi_1(t)^2 \, \der t \ \der s 
  \\
  &= 2\pi \delta
  \int_0^\infty \frac{\varphi_n(s)^2}{s}
  \left(1-\int_{c_+(s)}^\infty \xi_1(t)^2 \, \der t \right)
  \der s 
  \\
  &= 2\pi \delta
  \int_0^\infty \frac{\varphi_n(s)^2}{s}
  \left(1-\frac{\xi_1(c_+(s))^2}{2\sqrt{-E_1}} \right)
  \der s 
  \,.
\end{aligned} 
\end{equation*}
Here the inequality is due to Lemma~\ref{Lem2},  
the first equality employs the normalisation of~$\xi_1$  
and the last equality is due to the asymptotics~\eqref{asymptotics}. 
Since $c_+(s) \to \infty$ as $s \to \infty$, one has 
\begin{equation*} 
  %\lim_{n\to\infty}
  Q_2\big(\phi_n t \xi_1,\varphi_n \xi_1\big)
  \leq -\pi\delta 
  %\lim_{n\to\infty}
  \int_0^\infty \frac{\varphi_n(s)^2}{s} \der s 
  + o(1)
\end{equation*}
as $n \to \infty$.
It remains to compute
\begin{equation}\label{compute} 
  \int_0^\infty \frac{\varphi_n(s)^2}{s} \, \der s 
  = \mbox{$\frac{2}{3}$} \log n
  \,.
\end{equation}
In summary, 
\begin{equation*} 
  Q_2\big(\phi_n t \xi_1,\varphi_n \xi_1\big)
  \leq - c_1 \log n + o(1)
  \qquad \mbox{as} \qquad 
  n \to \infty
  \,.
\end{equation*}
where 
$c_1 := \mbox{$\frac{2}{3}$} \pi \delta$ is positive.

\noindent  
\fbox{$Q_2$, order $\eps^2$} 
Integrating by parts as above, we have 
\begin{equation*} 
\begin{aligned}
  Q_2[\phi_n t \xi_1] 
  &= \int_U \phi_n(s)^2 \,
  (t\xi_1(t))' (t\xi_1(t))' \, (1-2M(s) t + K(s)t^2) \, \der \Sigma \, \der t
  \\
  & \quad + \int_U W(t) \, 
  \phi_n(s)^2 \,
  t\xi_1(t) \, t\xi_1(t) \,
  (1-2 M(s) t + K(s) t^2) \, \der \Sigma \, \der t
  \\
  & \quad
  - E_1 \int_U  \phi_n(s)^2 \, t\xi_1(t) \, t\xi_1(t) \,
  (1-2 M(s) t + K(s) t^2) \, \der \Sigma \, \der t
  \\
  &= \int_U  \phi_n(s)^2 \, t \xi_1(t) (t\xi_1(t))'
  (2M(s) - 2K(s)t) \, \der \Sigma \, \der t
  \\
  & \quad -2 \int_U  \phi_n(s)^2 \, t \xi_1(t) \xi_1'(t) 
  (1-2M(s)t +K(s)t^2) \, \der \Sigma \, \der t
  \\
  & \quad + \int_{\Sigma} \phi_n(s)^2 
  \left[
  t\xi_1(t) (t\xi_1(t))' (1-2 M(s) t + K(s) t^2) 
  \right]_{t=-c_-(s)}^{t=c_+(s)} \, \der \Sigma
  \\
  &= \int_U \phi_n(s)^2 \, t^2 \xi_1(t)^2 \,  
  K(s) \, \der \Sigma \, \der t
  \\
  &\quad + \int_U  \phi_n(s)^2 \, \xi_1(t)^2  
  (1-4M(s)t +3K(s)t^2) \, \der \Sigma \, \der t
  \\
  & \quad + \int_{\Sigma} \phi_n(s)^2   
  \left[
  \xi_1(t)^2 (t-2 M(s) t^2 + K(s) t^3)
  + t^2\xi_1(t)^2 (M(s)-K(s)t)
  \right]_{t=-c_-(s)}^{t=c_+(s)} \, \der \Sigma
  \\
  & \quad + \int_{\Sigma} \phi_n(s)^2   
  \left[
  t\xi_1(t) (t\xi_1(t))' (1-2 M(s) t + K(s) t^2) 
  \right]_{t=-c_-(s)}^{t=c_+(s)} \, \der \Sigma
  \,.
\end{aligned}  
\end{equation*}
As above, we conclude that 
\begin{equation*} 
\begin{aligned}
  %\lim_{n\to\infty}
  Q_2[\phi_n t \xi_1] 
  &= %\lim_{n\to\infty}
  \int_U \phi_n(s)^2 \, \xi_1(t)^2 \,  
  (1-4M(s)t) \, \der \Sigma \, \der t
  + o(1)
  \\
  &= %\lim_{n\to\infty}
  \int_U \phi_n(s)^2 \, \xi_1(t)^2 \,  
  (1-2k_2(s)t) \, \der \Sigma \, \der t
  + o(1)
  \\
  &\leq %\lim_{n\to\infty}
  \int_U \phi_n(s)^2 \, \xi_1(t)^2 \, \der \Sigma \, \der t
  + o(1)
\end{aligned}  
\end{equation*}
as $n \to \infty$.
Here (using Lemma~\ref{Lem1})
\begin{equation*} 
\begin{aligned}
  \int_U \phi_n(s)^2 \, \xi_1(t)^2 \,  
  \der \Sigma \, \der t
  &\leq 2\pi C_\Sigma 
  \int_0^\infty \frac{\varphi_n(s)^2}{s^2} \, s
  \int_{-c_-(s)}^{c_+(s)} \xi_1(t)^2 \, \der t \ \der s 
  \\
  &= 2\pi C_\Sigma  
  \int_0^\infty \frac{\varphi_n(s)^2}{s}  
  \left(1-\int_{c_+(s)}^\infty \xi_1(t)^2 \, \der t \right)
  \der s 
  \\
  &= 2\pi C_\Sigma  
  \int_0^\infty \frac{\varphi_n(s)^2}{s} 
  \left(1-\frac{\xi_1(c_+(s))^2}{2\sqrt{-E_1}} \right)
  \der s 
  \,.
\end{aligned} 
\end{equation*}
Consequently,
\begin{equation*} 
  Q_2[\phi_n t \xi_1]  
  \leq c_2 \log n + o(1)
  \qquad \mbox{as} \qquad 
  n \to \infty
  \,.
\end{equation*}
where 
$c_2 := \mbox{$\frac{2}{3}$} 2\pi C_\Sigma$. 

\noindent 
\fbox{$Q$} 
Putting the results together, we finally arrive at
\begin{equation}\label{below}
\begin{aligned}
  Q[\psi_{n,\eps}] 
  &\leq (-2\eps c_1 + \eps^2 c_2)  \log n + o(1)
  \qquad \mbox{as} \qquad 
  n \to \infty
  \,.
\end{aligned}
\end{equation}
Obviously, it is possible to choose~$\eps$ positive 
and sufficiently small so that $Q[\psi_{n,\eps}]$
is negative for all sufficiently large~$n$.   
(Note that $\|\psi_{n,\eps}\| \to \infty$ as $n \to \infty$,
so the result~\eqref{below} does not contradict the fact
that~$Q$ is bounded from below.)
 
\medskip 
\noindent 
\fbox{$1 \mapsto \infty$} 
The argument above together with Theorem~\ref{Thm.ess.general}
demonstrates that there is at least one discrete eigenvalue
(below~$E_1$).
To realise that the same argument actually shows that
there is an infinite number (counting multiplicities) of discrete eigenvalues,
it is enough to notice that we have constructed a non-compact sequence
of trial functions. Indeed, $\{\psi_{n,\eps}\}_{n=2}^\infty$
certainly contains an infinite subsequence of functions 
with mutually disjoint supports.
\end{proof}

Theorem~\ref{Thm.main} from the introduction is
a combination of Theorems~\ref{Thm.ess.general} and~\ref{Thm.main.bis}
as well as of the observation that
the trial function from the proof of Theorem~\ref{Thm.main.bis}
``does not feel'' what happens on any compact subset of~$\Real^3$. 
In fact, the perturbation of~$\Sigma$ can be 
a conical surface as in~\cite{Egger-Kerner-Pankrashkin_2020}.  
At the same time, the essential spectrum of~$H$ 
is stable under changes of~$\Sigma$ on a compact set. 
The extra property that $\inf\sigma_\mathrm{ess}(H) = E_1$
in the last part of Theorem~\ref{Thm.main} follows from 
the general inequality $\inf\sigma_\mathrm{ess}(H) \geq E_1$
and the fact that there is an infinite number of discrete eigenvalues,
which can accumulate to the lowest point of the essential
spectrum only.

\subsection*{Acknowledgment}
We thank the anonymous referee for valuable comments.
D.K.\ was supported
by the EXPRO grant No.~20-17749X
of the Czech Science Foundation. 

%\newpage
%\vfill
%--------------%
% BIBLIOGRAPHY %
%--------------%
%
{\small
%\addcontentsline{toc}{section}{References}
\bibliography{bib}
\bibliographystyle{amsplain}
}

\end{document}